\documentclass[11pt,a4paper,reqno]{amsart}

\usepackage{times}
\usepackage{xspace}
\usepackage[inline]{enumitem}
\usepackage[numbers,sort&compress,longnamesfirst]{natbib}
\usepackage{geometry}
\usepackage[colorlinks=true, linkcolor=blue ,citecolor=blue]{hyperref}
\usepackage[capitalize]{cleveref}
\usepackage{amsfonts,comment}
\usepackage{amsmath,mathtools}
\usepackage{amssymb,latexsym}
\usepackage{color, graphicx}
\usepackage{todonotes}
\usepackage{algorithm}
\usepackage{algpseudocode}
\usepackage{subcaption}
\usepackage{tabu}
\usepackage{cleveref}

\theoremstyle{plain}

\newtheorem{theorem}{Theorem}[section]

\newtheorem{lemma}[theorem]{Lemma}
\newtheorem{examples}[theorem]{Examples}
\newtheorem{foo}[theorem]{Remarks}
\newtheorem{As}{Assumption}

\theoremstyle{definition}
\newtheorem{remark}[theorem]{Remark}
\newtheorem{example}[theorem]{Example}

\newcommand{\be}{\begin{equation}}
\newcommand{\ee}{\end{equation}}
\newcommand{\bea}{\begin{eqnarray}}
\newcommand{\eea}{\end{eqnarray}}
\newcommand{\beas}{\begin{eqnarray*}}
\newcommand{\eeas}{\end{eqnarray*}}

\newlist{myitems}{enumerate}{1}
\setlist[myitems]{label=\arabic*, font=\bfseries, resume}

\def\<{\langle}
\def\>{\rangle}

\def\\ud sum{\displaystyle\sum}

\def\ue{\mathrm e}
\def\ud{\mathrm d}

\numberwithin{equation}{section}
\numberwithin{figure}{section}
\numberwithin{table}{section}

\begin{document}

\title[A deep gradient flow method for option pricing]{A time-stepping deep gradient flow method for option pricing\\in (rough) diffusion models}

\author[A. Papapantoleon]{Antonis Papapantoleon}
\author[J. Rou]{Jasper Rou}

\address{Delft Institute of Applied Mathematics, EEMCS, TU Delft, 2628CD Delft, The Netherlands \& Department of Mathematics, School of Applied Mathematical and Physical Sciences, National Technical University of Athens, 15780 Zografou, Greece \& Institute of Applied and Computational Mathematics, FORTH, 70013 Heraklion, Greece}
\email{\href{mailto:a.papapantoleon@tudelft.nl}{a.papapantoleon@tudelft.nl}}

\address{Delft Institute of Applied Mathematics, EEMCS, TU Delft, 2628CD Delft, The Netherlands}
\email{\href{mailto:j.g.rou@tudelft.nl}{j.g.rou@tudelft.nl}}

\thanks{The authors are grateful to Emmanuil Georgoulis and Costas Smaragdakis for several helpful discussions and advice during the work on this project.
AP gratefully acknowledges the financial support from the Hellenic Foundation for Research and Innovation (H.F.R.I.) under the ``First Call for H.F.R.I. Research Projects to support Faculty members and Researchers and the procurement of high-cost research equipment grant'' (Project Number: 2152).
The authors also acknowledge the use of computational resources of the DelftBlue supercomputer, provided by the Delft High-Performance Computing Centre \cite{DHPC2022}.}



\date{}

\begin{abstract}
We develop a novel deep learning approach for pricing European options in diffusion models, that can efficiently handle high-dimensional problems resulting from Markovian approximations of rough volatility models.
The option pricing partial differential equation is reformulated as an energy minimization problem, which is approximated in a time-stepping fashion by deep artificial neural networks.
The proposed scheme respects the asymptotic behavior of option prices for large levels of moneyness, and adheres to \textit{a priori} known bounds for option prices.
The accuracy and efficiency of the proposed method is assessed in a series of numerical examples, with particular focus in the lifted Heston model.
\end{abstract}

\maketitle


\section{Introduction}

Stochastic volatility models have been popular in the mathematical finance literature because they allow to accurately model and reproduce the shape of implied volatility smiles for a single maturity.
They require though certain modifications, such as making the parameters time- or maturity-dependent, in order to reproduce a whole volatility surface; see \textit{e.g.} the comprehensive books by \citet{Gatheral_2012} or \citet{Bergomi_2016}.
The class of rough volatility models, in which the volatility process is driven by a \textit{fractional} Brownian motion, offers an attractive alternative to classical volatility models, since they allow to reproduce many stylized facts of asset and option prices with only a few (constant) parameters; see \textit{e.g.} the seminal articles by \citet{gatheral2018volatility} and \citet{bayer2016pricing}, and the recent volume by \citet{Bayer_etal_SIAM_2023}.
The presence of fractional Brownian motion in the dynamics of rough volatility models yields that the volatility process is not a semimartingale, and the model is not Markovian, which means that the simulation of the dynamics, or the pricing and hedging of derivatives, become challenging tasks.
This has attracted increasing attention from the mathematical finance community and has led to the creation of several innovative methods to tackle these challenges.

Markovian approximations of fractional processes is a popular method for tackling the challenges arising in these models, since they allow to relate rough volatility models to their classical counterparts. 
\citet{abi2019multifactor,abi2019markovian} have designed multifactor stochastic volatility models approximating rough volatility models and enjoying a Markovian structure. 
\citet{zhu2021markovian} have showed that the rough Bergomi model can be well-approximated by the forward-variance Bergomi model with appropriately chosen weights and mean-reversion speed parameters, which has the Markovian property.
\citet{bayer2023markovian,bayer2023weak} have studied Markovian approximations of stochastic Volterra equations using an $N$-dimensional diffusion process defined as the solution to a system of (ordinary) stochastic differential equations. 
\citet{cuchiero2020generalized} have provided existence, uniqueness and approximation results for Markovian lifts of affine rough volatility models of general (jump) diffusion type.
Moreover, \citet{bonesini2023elephant} have proposed a theoretical framework that exploits convolution kernels to transform a Volterra path-dependent (non-Markovian) stochastic process into a standard (Markovian) diffusion process.

Novel simulation schemes have also been developed in order to speed up Monte Carlo methods for rough volatility models.
\citet{Bennedsen_Lunde_Pakkanen_2017} have developed a hybrid simulation scheme for Brownian semistationary processes, and applied this to the rough Bergomi model.
This method was further refined by \citet{McCrickerd_Pakkanen_2018} who developed variance reduction techniques for conditionally log-normal models, and by \citet{Fukasawa_Hirano_2021} who optimized the use of random numbers through orthogonal projections.
Moreover, \citet{Gatheral_2022} has developed hybrid quadratic-exponential schemes for the simulation of rough affine forward variance models, while \citet{bayer2024efficient} have developed efficient and accurate simulation schemes for the rough Heston model based on low-dimensional Markovian representations.

Deep learning and neural network methods have also been applied in the field of rough stochastic volatility modeling.
\citet{horvath2021deep2} presented a neural network-based method that performs the calibration task for the full implied volatility surface in (rough) volatility models.
\citet{jacquier2023deep} introduced a deep learning–based algorithm to evaluate options in affine rough stochastic volatility models.
\citet{jacquier2023random} constructed a deep learning-based numerical algorithm to solve path-dependent partial differential equations arising in the context of rough volatility. 

Many other methods have been utilized and applied to rough stochastic volatility models. 
Let us mention here, among others, that \citet{el2019characteristic} have computed the characteristic function of the log‐price in the rough Heston model, thus opening the door for the application of Fourier transform methods for option pricing. 
\citet{Bayer_BenHammouda_Tempone_2020} proposed the application of adaptive sparse grids and quasi Monte Carlo methods for option pricing in the rough Bergomi model, while \citet{Bayer_Friz_Gassiat_Martin_Stemper_2020} have applied Hairer's theory of regularity structures in order to analyze rough volatility models.
In addition, \citet{bayer2019short, Horvath_Jacquier_Lacombe_2019} and \citet{Jacquier_Pannier_2022}, among others, have studied asymptotic properties of rough volatility models in various regimes.

The connection between partial differential equations (PDEs) and option pricing is deep and well-studied in the literature.
Several methods have been developed in order to price options in diffusion or stochastic volatility models, especially in high-dimensional settings, using sparse grids, asymptotic expansions, finite difference and finite elements methods; see \textit{e.g.} \citet{Reisinger_Wittum_2007,DURING20124462,Hilber_Matache_Schwab_2005} and the comprehensive book by \citet{hilber2013computational}.

The recent advances in deep learning methods have also created an intense interest in the solution of PDEs by deep learning and neural network methods.
\citet{sirignano2018dgm} proposed to solve high-dimensional PDEs by approximating the solution with a deep neural network which is trained to satisfy the differential operator, initial condition, and boundary conditions.
\citet{raissi2019physics} introduced physics informed neural networks, \textit{i.e.} neural networks that are trained to solve supervised learning tasks while respecting any given law of physics described by general nonlinear PDEs.
\citet{van2022optimally} discussed the choice of a norm in the loss function for the training of neural networks.
\citet{han2018solving} introduced a deep learning-based approach that can handle general high-dimensional parabolic PDEs by reformulating the PDE using backward stochastic differential equations and approximating the gradient of the unknown solution by neural networks.
\citet{yu2018deep} proposed a deep learning-based method for numerically solving variational problems, particularly the ones that arise from PDEs.
\citet{liao2019deep} proposed a method to deal with the essential boundary conditions encountered in the deep learning-based numerical solvers for PDEs.
\citet{georgoulis2023discrete} considered the approximation of initial and boundary value problems involving high-dimensional, dissipative evolution PDEs using a deep gradient flow neural network framework. 
In similar spirit, \citet{park2023deep} proposed a deep learning-based minimizing movement scheme for approximating the solutions of PDEs.
The work of \cite{georgoulis2023discrete} was followed up by \citet{georgoulis2024deep}, who developed an implicit-explicit deep minimizing movement method for pricing European basket options written on assets that follow jump-diffusion dynamics.
Let us refer the reader to the forthcoming book by \citet{jentzen2023mathematical}, who provide an introduction to the topic of deep learning algorithms and a survey of the (vast) literature.

The starting point for our research is the existence of a Markovian approximation for a rough volatility model or, more generally, another Markovian model that describes empirical data well; see \textit{e.g.} \citet{Romer_2022} and \citet{abi2024volatility}.
Inspired by \citet{georgoulis2023discrete}, we consider the option pricing PDE, introduce a time-stepping discretization of the time derivative, and reformulate the spatial derivatives as an energy minimization problem; the main difference with \citet{georgoulis2023discrete} is the presence of the drift (first order) term, which we treat in an explicit fashion in the discretization.
This gradient flow formulation of the PDE gives rise to a suitable cost function for the training, via stochastic gradient descent, of deep neural networks that approximate the solution of the PDE.
The main advantages of this method are twofold: on the one hand, the neural network from the previous time step is a good initial guess for the neural network optimization in the current time step, thus reducing the number of training stages necessary.
On the other hand, the energy formulation allows to consider only first-order derivatives, thus reducing the training time in each time step / training stage significantly.
In addition, we utilize available information about the qualitative behavior of option prices in order to assist the training of the neural network. 
More specifically, the architecture of the neural network respects the asymptotic behavior of option prices for large levels of moneyness, and adheres to \textit{a priori} known bounds for option prices.
The empirical results show that the proposed method is accurate, with a small overall error in the order of $10^{-3}$ and comparable with other deep learning methods, while the training times are significantly smaller compared to similar methods and do not grow significantly with the dimension of the underlying space.

This article is organized as follows: in Section \ref{sec:problem}, we revisit the option pricing PDE for Markovian diffusion models.
In Section \ref{sec:TDGF}, we explain how we develop our method for solving option pricing PDEs, with focus on the time discretization and energy formulation.
In Section \ref{sec:examples}, we revisit classical models for asset prices, with particular focus in the lifted Heston model.
In Section \ref{sec:implementation_details}, we explain certain details for the design and implementation of the method. 
In Section \ref{sec:results}, we present and discuss the outcome of our numerical experiments, while Section \ref{sec:conclusion} concludes this article.


\section{Problem formulation}
\label{sec:problem}

Let $S$ denote the price process of a financial asset that evolves according to a (Markovian) diffusion model, and consider a European vanilla derivative on $S$ with payoff $\Psi(S_T)$ at maturity time $T>0$.
Using the fundamental theorem of asset pricing and the Feynman--Kac formula, the price of this derivative can be written as the solution to a PDE in these models. 
Indeed, let $u: [0,T] \times \Omega \to \mathbb{R}$ denote the price of this derivative, where $\Omega \subseteq \mathbb{R}^{n+1}$ and $t$ denotes the time to maturity. 
Then, $u$ solves the general PDE
\begin{equation}
\label{eq:non_linear_pde}
\begin{aligned}
\begin{cases}
    u_t(t,x) + \mathcal{A} u(t,x) + ru(t,x) = 0, \quad & (t,x) \in (0,T] \times \Omega,\\
    u(0,x) = \Psi(x), \quad & x \in \Omega,
\end{cases} 
\end{aligned}
\end{equation}
where $\mathcal{A}$ is a second order differential operator of the form
\begin{equation}
\label{eq:generator}
    \mathcal{A} u = - \sum_{i,j=0}^n a^{ij} \frac{\partial^2 u}{\partial x_i \partial x_j} + \sum_{i=0}^{n} \beta^i \frac{\partial u}{\partial x_i}.    
\end{equation}
The coefficients $a^{ij}, \beta^i$ of the generator $\mathcal{A}$ are directly related to the dynamics of the stochastic process $S$ and can, in general, depend on the time and the spatial variables. 
Their particular form will be clarified in the applications later; see \cref{sec:examples}.
Moreover, for the time being, let us also prescribe homogeneous Dirichlet boundary conditions on $\partial \Omega$.

We will later rewrite the PDE \eqref{eq:non_linear_pde} as an energy minimization problem, therefore we need to split the operator in a symmetric and an asymmetric part. 
The symmetric part will then be associated to a Dirichlet energy.
Using the product rule, we can rewrite $\mathcal{A}$ as follows
\begin{align*}
    \mathcal{A} u & = \sum_{i=0}^n \beta^i \frac{\partial u}{\partial x_i} - \sum_{i,j=0}^n \frac{\partial}{\partial x_j} \left( a^{ij} \frac{\partial u}{\partial x_i} \right) + \sum_{i,j=0}^n \frac{\partial a^{ij}}{\partial x_j} \frac{\partial u}{\partial x_i} \\
    & = \sum_{i=0}^n \left( \beta^i + \sum_{j=0}^n \frac{\partial a^{ij}}{\partial x_j} \right) \frac{\partial u}{\partial x_i} - \sum_{i,j=0}^n \frac{\partial}{\partial x_j} \left( a^{ij} \frac{\partial u}{\partial x_i} \right).
\end{align*}
Define $b^i = \beta^i + \sum_{j=0}^n \frac{\partial a^{ij}}{\partial x_j}$. 
Then, $\mathcal{A}$ takes the form
\begin{equation}
\label{eq:operator_form}
    \mathcal{A} u = - \nabla \cdot \left( A \nabla u \right) + \mathbf{b} \cdot \nabla u,   
\end{equation}
where
\begin{equation}
\label{eq:coefficients}
A = \begin{bmatrix}
    a^{00} & a^{10} & \dots & a^{n0} \\
    a^{01} & a^{11} & \dots & a^{n1} \\
    \vdots & \vdots & \ddots & \vdots \\
    a^{0n} & a^{1n} & \dots & a^{nn}
\end{bmatrix}
\quad \textrm{and} \quad
\mathbf{b} = \begin{bmatrix}
    b^0 \\
    b^1 \\
    \vdots \\
    b^n
\end{bmatrix}.
\end{equation}
Let us point out that, the order of the derivatives with respect to $x_i$ and $x_j$ in the second derivative term can be interchanged in \eqref{eq:generator}, hence we can always choose $a^{ij} = a^{ji}$. 
Therefore, $A$ is a symmetric matrix.


\section{Energy minimization and a deep gradient flow method}
\label{sec:TDGF}

Let us now explain how we will develop our deep neural network method for solving the PDE \eqref{eq:non_linear_pde}.
We proceed in three steps.
First, we discretize the time derivative in Subsection \ref{sec:time} and introduce a time-stepping scheme. 
Second, we rewrite the solution of the PDE as an energy minimization problem in Subsection \ref{sec:variational}, which gives rise to a suitable cost function. 
Third, we approximate the minimizer by a neural network, which is trained using stochastic gradient descent, in Subsection \ref{sec:neural}


\subsection{Time discretization}
\label{sec:time}

A key ingredient of our method is that we do not use time as an input variable for the neural network, thus training a global space-time network for solving the PDE, as in \citet{sirignano2018dgm} and \citet{raissi2019physics}.
Instead, following \citet{georgoulis2023discrete}, we wish to train a neural network for each time step.
The main insight is that the neural network from the previous time step is a good initial guess for the neural network optimization in the current time step, thus reducing the number of training stages necessary.

Let us divide the time interval $(0,T]$ into $N$ equally spaced intervals $(t_{k-1},t_k]$, with $h = t_k - t_{k-1} = \frac{1}{N}$ for $k=0,1,\dots,N$. 
Let $U^k$ denote the approximation to the solution of the PDE $u(t_k)$ at time step $t_k$, using the following time discretization scheme
\[
\begin{aligned}
    \frac{U^k - U^{k-1}}{h} - \nabla \cdot \left( A \nabla U^k \right) + \mathbf{b} \cdot \nabla U^{k} + r U^k & = 0, \\
    U^0 & = \Psi.
\end{aligned}
\]
In order to be able to rewrite the PDE as an energy minimization problem, we wish to treat the asymmetric part as a constant function. 
Therefore, we substitute the function $U^k$ in the asymmetric part with its value from the previous time point, and thus consider the backward differentiation scheme
\begin{equation}
\label{eq:discretization}
\begin{aligned}
\frac{U^k - U^{k-1}}{h} - \nabla \cdot \left( A \nabla U^k \right) + F \left( U^{k-1} \right) + r U^k & = 0, \\
    U^0 & = \Psi,
\end{aligned}
\end{equation}
where $F(u) = \mathbf{b} \cdot \nabla u$.
Convergence analysis of this scheme appears in \citet{Akrivis_Crouzeix_Makridakis_1999}.

\begin{remark}
Let us point out that non-uniform grids can also be used for the discretization of time, in order to employ a finer grid initially when the solution is less smooth. 
The numerical experiments in \cref{sec:examples} show that this is not necessary in the examples we consider.
\end{remark}


\subsection{Variational formulation}
\label{sec:variational}

We would like now to associate the discretized PDE in \eqref{eq:discretization} with an energy functional, in order to define later a suitable cost function for the training of the neural network.
In other words, we would like to find an energy functional $I(u)$ such that $U^k$ is a critical point of $I$. 
Let us rewrite \eqref{eq:discretization} as follows:
\begin{equation}
\label{eq:discretization2}
\left( U^k - U^{k-1} \right) + h \left( - \nabla \cdot \left( A \nabla U^k \right) + r U^k + F \left( U^{k-1} \right) \right) = 0, 
    \quad U^0 = \Psi.
\end{equation}
Then, we have the following result.

\begin{lemma}\label{lem_var-form}
Assume that $A$ is symmetric and positive semi-definite.
Then, $U^k$ solves \eqref{eq:discretization2} if and only if $U^k$ minimizes
\[
    \begin{aligned}
        I^k(u) = & \frac{1}{2} \left \Vert u - U^{k-1} \right \Vert_{L^2(\Omega)}^2 + h \left( \int_{\Omega} \frac{1}{2} \left( \left( \nabla u \right)^\mathsf{T} A \nabla u + r u^2 \right) + F \left( U^{k-1} \right) u \ud x
        \right).
    \end{aligned}
\]
\end{lemma}

\begin{proof}
Let $v$ be a smooth function that is zero on the boundary.
Then, for all such $v$, consider the function
\[
\begin{aligned}
i^k(\tau) &=  I^k \left( U^k + \tau v \right) = \int_{\Omega} \frac{1}{2} \left( U^k + \tau v - U^{k-1} \right)^2 \ud x \\
&\quad + h \int_{\Omega} \frac{1}{2} \left( \left( \nabla \left( U^k + \tau v \right) \right)^\mathsf{T} A \nabla \left( U^k + \tau v \right) + r \left( U^k + \tau v \right)^2 \right) + F \left( U^{k-1} \right) \left( U^k + \tau v \right) \ud x,
\end{aligned}
\]
for $\tau \in \mathbb{R}$. 
$U^k$ minimizes $I^k$ if and only if $\tau = 0$ minimizes $i^k$. 
Therefore, we get
\begin{align*}
    0 &=  (i^k) '(0) = \left[ \int_{\Omega} \left( U^k + \tau v - U^{k-1} \right) v \ud x \right. \\
    &\quad + \left. h \int_{\Omega} \frac{1}{2} \left( \left( \nabla v \right)^\mathsf{T} A \nabla \left( U^k + \tau v \right) + \left( \nabla \left( U^k + \tau v \right) \right)^\mathsf{T} A \nabla v \right) + r \left( U^k + \tau v \right) v + F \left( U^{k-1} \right) v \ud x \right]_{\tau=0} \\ 
    &= \int_{\Omega}  \left( U^k - U^{k-1} \right) v + \frac{h}{2} \left( \left( \nabla v \right)^\mathsf{T} A \nabla U^k + \left( \nabla U^k \right)^\mathsf{T} A \nabla v \right) + h r U^k v  + h F \left( U^{k-1} \right) v \ud x \\
    &= \int_{\Omega} \left( U^k - U^{k-1} \right) v +  h  A \nabla U^k \cdot \nabla v + h \left( r U^k + F \left( U^{k-1} \right) \right) v \ud x \\
    &= \int_{\Omega}  \left( \left( U^k - U^{k-1} \right) + h \left( - \nabla \cdot \left( A \nabla U^k \right) + r U^k + F \left( U^{k-1} \right) \right) \right) v \ud x,
\end{align*}
where in the second to last equality we used that $A$ is symmetric, and in the last equality we used the divergence theorem. 
This equality holds for all $v$ if and only if \eqref{eq:discretization2} holds.
The second derivative is positive; indeed, 
\[
\left( i^k \right)''(\tau) = (1+rh) \int_{\Omega} v^2 \ud x + h \int_{\Omega} \left( \nabla v \right)^\mathsf{T} A \nabla v \ud x > 0,
\]
since $A$ is positive semi-definite. 
Therefore, $\tau=0$ is indeed the minimizer.
\end{proof}


\subsection{Neural network approximation}
\label{sec:neural}

Using the results of the previous subsection, we define the following cost (loss) function for training the neural network:
\begin{align*}
L^k(u) = \frac{1}{2} \left \Vert u - U^{k-1} \right \Vert_{L^2(\Omega)}^2 
    + h  \int_{\Omega} \mathcal{E}[u]\ud x + h  \int_{\Omega} F \left( U^{k-1} \right) u \ud x,
\end{align*}
where
\begin{align*}
\mathcal{E}[u] = \frac{1}{2} \left( \left( \nabla u \right)^\mathsf{T} A \nabla u + r u^2 \right),
\end{align*}
denotes the respective Dirichlet energy.
    
Let $f^k(\mathbf{x}; \theta)$ denote a neural network approximation of $U^k$ with trainable parameters $\theta$. 
Applying a Monte Carlo approximation to the integrals, the discretized cost functional takes the form
\begin{align*}
L^k (\theta ; \mathbf{x}) 
    & = \frac{|\Omega|}{2M} \sum_{i=1}^{M} \left( f^k(\mathbf{x}_i;\theta) - f^{k-1}(\mathbf{x}_i) \right)^2 + h N^k \left( \theta ; \mathbf{x} \right),
\end{align*}
where
\begin{align*}
N^k \left( \theta ; \mathbf{x} \right) 
    &=  \frac{|\Omega|}{M} \sum_{i=1}^{M} \left[ \frac{1}{2} \left( \left( \nabla f^k(\mathbf{x}_i; \theta) \right)^\mathsf{T} A \nabla f^k(\mathbf{x}_i; \theta) + r \left( f^k({\mathbf{x}_i}; \theta) \right) ^2 \right) + \left( \mathbf{b} \cdot \nabla f^{k-1}(\mathbf{x}_i) \right) f^k(\mathbf{x}_i; \theta) \right].
\end{align*}
Here, $M$ denotes the number of samples $\mathbf{x}_i$.

In order to minimize this cost function, we use a stochastic gradient descent-type algorithm, \textit{i.e.} an iterative scheme of the form: 
\[
    \theta_{n+1} = \theta_n - \alpha_n \nabla_{\theta} L^k(\theta_n; \mathbf{x});
\] 
more specifically, we will use the Adam algorithm \cite{kingma2014adam}.
The hyper-parameter $\alpha_n$ is the step-size of our update, called the learning rate. 
An overview of the Time Deep Gradient Flow (TDGF) method appears in Algorithm \ref{alg:time_DNM}.

\begin{algorithm}
\caption{Time Deep Gradient Flow method}\label{alg:time_DNM}
\begin{algorithmic}[1]
\State Initialize $\theta_0^0$.
\State Set $f^0(\mathbf{x};\theta) = \Psi(\mathbf{x})$.
\For{each time step $k = 1,\dots,N$}
\State Initialize $\theta_0^k = \theta^{k-1}$.
\For{each sampling stage $n$}
\State Generate random points $\mathbf{x}_i$ for training.
\State Calculate the cost functional $L^k(\theta_n^k; \mathbf{x}_i)$ for the sampled points.
\State Take a descent step $\theta_{n+1}^k = \theta_n^k - \alpha_n \nabla_{\theta} L^k(\theta_n^k; \mathbf{x}_i)$.
\EndFor
\EndFor
\end{algorithmic}
\end{algorithm}


\section{Examples}
\label{sec:examples}

Let us now outline the different models to which we will apply our method, and derive the appropriate form of the generator in the option pricing PDE for each of these models.
In particular, we will treat the Black--Scholes model in Subsection \ref{sec:BS}, the Heston model in Subsection \ref{sec:heston}, and the lifted Heston model in Subsection \ref{sec:lifted_Heston}.
Let us point out that, straightforward but tedious computations, yield that the operator $A$ satisfies the assumptions of Lemma \ref{lem_var-form} for all three examples.


\subsection{Black--Scholes model}
\label{sec:BS}

In the \citet{black1973pricing} model, the dynamics of the stock price $S$ follow a geometric Brownian motion:
\[
\ud S_t = r S_t \ud t + \sigma S_t \ud W_t, \quad S_0 > 0,
\]
with $r, \sigma \in \mathbb{R}_+$ the risk free rate and the volatility respectively. 
The generator corresponding to these dynamics, in the form \eqref{eq:generator}, equals
\[
\mathcal{A} u = -\frac{1}{2} \sigma^2 x^2 \frac{\partial^2 u}{\partial x^2} - r x \frac{\partial u}{\partial x},
\]
and, applying the chain rule, we have
\[
\mathcal{A} u = - \frac{\partial}{\partial x} \left( \frac{1}{2} \sigma^2 x^2 \frac{\partial u}{\partial x} \right) + \sigma^2 x \frac{\partial u}{\partial x} - r x \frac{\partial u}{\partial x}.
\]
Therefore, the operator $\mathcal{A}$ takes the form \eqref{eq:operator_form} with the coefficients in \eqref{eq:coefficients} provided by 
\begin{align*}
    a &= \frac{1}{2} \sigma^2 x^2, \\
    b &= (\sigma^2 - r) x.
\end{align*}


\subsection{Heston model}
\label{sec:heston}

The \citet{heston1993closed} model is a popular stochastic volatility model with dynamics
\[
\begin{aligned}
\ud S_t & = r S_t \ud t + \sqrt{V_t} S_t \ud W_t, \quad & S_0 > 0, \\
\ud V_t & = \lambda ( \kappa - V_t) \ud t + \eta \sqrt{V_t} \ud B_t, \quad & V_0 > 0.
\end{aligned}
\]
Here $V$ is the variance process, $W, B$ are correlated (standard) Brownian motions, with correlation coefficient $\rho$, and $\lambda, \kappa, \eta \in \mathbb{R}_+$. 
The option pricing PDE \eqref{eq:non_linear_pde} in this model has two spatial variables $(x,v)$, corresponding to the asset price and the variance, and the generator corresponding to these dynamics, in the form \eqref{eq:generator}, equals
\[
\mathcal{A} u = - r x \frac{\partial u}{\partial x} - \lambda (\kappa - v) \frac{\partial u}{\partial v} - \frac{1}{2} x^2 v \frac{\partial^2 u}{\partial x^2} - \frac{1}{2} \eta^2 v \frac{\partial^2 u}{\partial v^2} - \rho \eta x v \frac{\partial^2 u}{\partial x \partial v}.
\]
Applying the chain rule again, we get
\[
\begin{aligned}
\mathcal{A} u = & - r x \frac{\partial u}{\partial x} - \lambda (\kappa - v) \frac{\partial u}{\partial v} - \frac{\partial}{\partial x} \left( \frac{1}{2} x^2 v \frac{\partial u}{\partial x} \right) + x v \frac{\partial u}{\partial x} - \frac{\partial}{\partial v} \left( \frac{1}{2} \eta^2 v \frac{\partial u}{\partial v} \right) + \frac{1}{2} \eta^2 \frac{\partial u}{\partial v}\\
& - \frac{\partial}{\partial x} \left( \frac{1}{2} \rho \eta x v \frac{\partial u}{\partial v} \right) + \frac{1}{2} \rho \eta v \frac{\partial u}{\partial v} - \frac{\partial}{\partial v} \left( \frac{1}{2} \rho \eta x v \frac{\partial u}{\partial x} \right) + \frac{1}{2} \rho \eta x \frac{\partial u}{\partial x}.
\end{aligned}
\]
Therefore, the operator $\mathcal{A}$ takes the form \eqref{eq:operator_form} with the coefficients in \eqref{eq:coefficients} provided by 
\begin{align*}
    a^{00} &= \frac{1}{2} x^2 v, \\
    a^{10} = a^{01} &= \frac{1}{2} \rho \eta x v, \\
    a^{11} & = \frac{1}{2} \eta^2 v, \\
    b^0 &= \left(-r + v + \frac{1}{2} \rho \eta \right) x, \\
    b^1 &= \lambda (v - \kappa) + \frac{1}{2} \eta^2 + \frac{1}{2} \rho \eta v.
\end{align*} 


\subsection{Lifted Heston model}
\label{sec:lifted_Heston}

The Heston model requires difficult modifications, such as making the parameters time dependent, in order to adequately calibrate the model to financial market data.
The rough Heston model, in which the volatility process is driven by a \textit{fractional} Brownian motion, is an alternative where only a few (constant) parameters are required in order to adequately calibrate the model; see \citet{gatheral2018volatility} and \citet{bayer2016pricing}.
The dynamics of the variance process in the rough Heston model are described by
\[
    V_t = V_0 + \int_0^t K(t-s)  \left( \lambda \left( \kappa - V_s \right) \ud s + \eta \sqrt{V_s} \ud B_s \right),
\]
where $K(t) = \frac{t^{H-\frac{1}{2}}}{\Gamma \left( H + \frac{1}{2} \right) }$ is the fractional kernel and $H \in \left( 0,\frac{1}{2} \right)$ is the Hurst index. 
This model is not Markovian, due to the presence of the fractional Brownian motion in the dynamics, hence the pricing and hedging of derivatives becomes a challenging task.
In particular, we do not have access to a PDE of the form \eqref{eq:non_linear_pde} for the rough Heston model.

The \textit{lifted} Heston model is a Markovian approximation of the rough Heston model, see \citet{abi2019multifactor}, that is interesting as a model in its own right; see also \citet{abi2019lifting}.
The dynamics of the asset price in the lifted Heston model are the following:
\begin{align}
\ud S_t & = r S_t \ud t + \sqrt{V_t^n} S_t \ud W_t, \quad S_0 > 0, \nonumber \\
\ud V_t^{n,i} & = - \left( \gamma_i^n V_t^{n,i} + \lambda V_t^n \right) \ud t + \eta \sqrt{V_t^n} \ud B_t,  \quad V_0^{n,i} = 0, \label{eq:Vni} \\
V_t^n & = g^n(t) + \sum_{i=1}^n c_i^n V_t^{n,i},  \label{eq:Vn} \\
g^n(t) & = V_0 + \lambda \kappa \sum_{i=1}^n c_i^n \int_0^t \ue^{- \gamma_i^n (t-s)} \ud s, \nonumber
\end{align}
where $\lambda, \eta, c_i^n, \gamma_i^n, V_0, \kappa \in \mathbb{R}_+$, and $W,B$ are two correlated (standard) Brownian motions with correlation coefficient $\rho$. 
The variance factors $V^{n,i}$ start from zero and share the same one-dimensional Brownian motion $B$.

The option pricing PDE \eqref{eq:non_linear_pde} has $n+1$ spatial variables $(x,v_1,\dots,v_n)$ in the lifted Heston model, corresponding to the asset price $S$ and the $n$ variance processes $V^{n,i}$ in \eqref{eq:Vni}.
The next section explains the sampling of the process $V^n$ in \eqref{eq:Vn}.

Let us now derive the generator for this system of SDEs.
We will denote by $V_i=V^{n,i}$, $f_{S} = \frac{\partial f}{\partial S}$ and $f_{V_i} = \frac{\partial f}{\partial V^{n,i}}$, and suppress the dependence on time for simplicity.
Let $f \left( S,V_1,\dots,V_n \right): \mathbb{R}^{n+1} \to \mathbb{R}$ be a $\mathcal{C}^2$-function. 
Using Itô's formula, we have that
\[
\begin{aligned}
\ud f \left( S,V_1,\dots,V_n \right) 
    &= f_{S} \ud S + \sum_{i=1}^n f_{V_i} \ud V_i + \frac{1}{2} f_{SS} \ud \left \langle S , S \right \rangle + \sum_{i=1}^n f_{S V_i} \ud \left \langle S, V_i \right \rangle + \frac{1}{2} \sum_{i,j=1}^n f_{V_i V_j} \ud \left \langle V_i, V_j \right \rangle\\
    &= f_{S} r S \ud t - \sum_{i=1}^n f_{V_i} \left(\gamma_i^n V_i + \lambda V_t^n \right) \ud t + \frac{1}{2} f_{SS} V_t^n S^2 \ud t + \sum_{i=1}^n f_{S V_i} \eta \rho V_t^n S \ud t \\
    &\quad + \frac{1}{2} \sum_{i,j=1}^n f_{V_i V_j} \eta^2 V_t^n \ud t + \textrm{(local) martingale}.
\end{aligned}
\]
Hence, the generator $\tilde{\mathcal A}$ of this system takes the form: 
\[
\tilde{\mathcal{A}} f = - r x f_{x} + \sum_{i=1}^n \left( \gamma_i^n v_i + \lambda V_t^n \right) f_{v_i} - \frac{1}{2} V_t^n x^2 f_{xx} - \eta \rho V_t^n x \sum_{i=1}^n f_{x v_i} - \frac{\eta^2}{2} V_t^n \sum_{i,j=1}^n f_{v_i v_j}.
\]
In order to bring this generator in the form \eqref{eq:generator}, we need to switch to the time to maturity, denoted now $\tau = T -t$.
We have to be particularly careful in this case, because the function $g^n$ in $V^n$ depends on the time $t$. 
Let us denote
\[
    \tilde V_{\tau}^n := g^n(T-\tau) + \sum_{i=1}^n c_i^n V_{\tau}^{n,i}.
\]
Then, applying the chain rule again, we have that
\begin{align*}
\mathcal{A} u 
    &=  - r x u_{x} + \sum_{i=1}^n \left(\gamma_i^n v_i + \lambda \tilde V_{\tau}^n \right) u_{v_i} - \left( \frac{1}{2} \tilde V_{\tau}^n x^2 u_{x} \right)_x + \tilde V_{\tau}^n x u_{x} - \frac{1}{2} \sum_{i=1}^n \left( \eta \rho \tilde V_{\tau}^n x u_{v_i} \right)_x \\
    &\quad + \frac{1}{2} \eta \rho \tilde V_{\tau}^n \sum_{i=1}^n u_{v_i} - \frac{1}{2} \sum_{i=1}^n \left( \eta \rho \tilde V_{\tau}^n x u_x \right)_{v_i} + \frac{1}{2} \eta \rho x \sum_{i=1}^n c_i^n u_x - \sum_{i,j=1}^n \left( \frac{\eta^2}{2} \tilde V_{\tau}^n u_{v_i} \right)_{v_j} \\ 
    &\quad + \frac{\eta^2}{2} \sum_{i,j=1}^n c_j^n u_{v_i} \\ 
    &= \left( \tilde V_{\tau}^n x - r x + \frac{1}{2} \eta \rho x \sum_{i=1}^n c_i^n \right) u_{x} + \sum_{i=1}^n \left(\gamma_i^n v_i + \lambda \tilde V_{\tau}^n + \frac{1}{2} \eta \rho \tilde V_{\tau}^n + \frac{\eta^2}{2} \sum_{j=1}^n c_j^n \right) u_{v_i} \\ 
    &\quad - \left( \frac{1}{2} \tilde V_{\tau}^n x^2 u_{x} \right)_x - \frac{1}{2} \sum_{i=1}^n \left( \eta \rho \tilde V_{\tau}^n x  u_{v_i} \right)_x - \frac{1}{2} \sum_{i=1}^n \left( \eta \rho \tilde V_{\tau}^n x u_x \right)_{v_i} - \sum_{i,j=1}^n \left( \frac{\eta^2}{2} \tilde V_{\tau}^n  u_{v_i} \right)_{v_j}.
\end{align*}
Therefore, the operator $\mathcal{A}$ takes the form \eqref{eq:operator_form} with the coefficients in \eqref{eq:coefficients} provided by 
\[
\begin{aligned}
a^{00} &= \frac{1}{2} \tilde V_{\tau}^n x^2,\\
a^{i0} = a^{0i} &= \frac{1}{2} \eta \rho \tilde V_{\tau}^n x, \quad & i=1,\dots,n, \\
a^{ij} &= \frac{\eta^2}{2} \tilde V_{\tau}^n, \quad & i,j=1,\dots,n, \\
b^0 &= \left( \tilde V_{\tau}^n - r + \frac{1}{2} \eta \rho \sum_{i=1}^n c_i^n \right) x, \\
b^i &= \left(\gamma_i^n v_i + \lambda \tilde V_{\tau}^n + \frac{1}{2} \eta \rho \tilde V_{\tau}^n + \frac{\eta^2}{2} \sum_{j=1}^n c_j^n \right), \quad & i=1,\dots,n. 
\end{aligned}
\]


\section{Implementation details}
\label{sec:implementation_details}

Let us now describe some details about the design of the neural network architecture and the implementation of the numerical method.
Motivated by \citet{georgoulis2024deep}, we would like to use information about the option price in order to facilitate the training of the neural network. 
On the one hand, using the lower no-arbitrage bound of a call option, \textit{i.e.} that $u(t, x) \geq \left(x - K \ue^{-rt} \right)^+$, the neural network will learn the difference between the option price and this bound, instead of the option price itself.
This difference turns out to be an easier function to approximate than the option price. 

On the other hand, after some large level of asset price / moneyness, the option price will grow linearly with the asset price / moneyness. 
Let us denote this level by $x_p$.
Therefore, instead of letting the neural network learn the option price at points larger than $x_p$, we will approximate this price by adding the difference between $x_p$ and this point to the option value at $x_p$.
In other words, 
\[
u(x_p + y; \theta) = u(x_p ; \theta) + y, \quad y>0.
\]
This is consistent with the imposition of Dirichlet boundary conditions on $\Omega$.

Moreover, we will consider the moneyness $\frac{S}{K}$ as an input variable instead of treating the asset price $S$ and the strike price $K$ separately, in order to reduce the size of the parameter space.\par

The architecture of the neural network for the TDGF method is inspired by the Deep Galerkin method (DGM) of \citet{sirignano2018dgm}, hence we call the hidden layers ``DGM layers''.
Overall, we set the following:
\begin{align*}
\text{input}\qquad    \quad
    X^0 & = \sigma_1 \left( W^1 \mathbf{x} + b^1 \right),  \\
    \text{ DGM layer}\\
    |\quad Z^l & = \sigma_1 \left( U^{z,l} \mathbf{x} + W^{z,l} X^{l-1} + b^{z,l} \right), \quad & l=1,\dots,L, \\
    |\quad G^l & = \sigma_1 \left( U^{g,l} \mathbf{x} + W^{g,l} X^{l-1} + b^{g,l} \right), \quad & l=1,\dots,L, \\
    |\quad R^l & = \sigma_1 \left( U^{r,l} \mathbf{x} + W^{r,l} X^{l-1} + b^{r,l} \right), \quad & l=1,\dots,L, \\
    |\quad\! H^l & = \sigma_1 \left( U^{h,l} \mathbf{x} + W^{h,l} \left( X^{l-1} \odot R^l \right) + b^{h,l} \right), \quad & l=1,\dots,L, \\
    \lfloor\quad\! X^l & = \left( 1 - G^l \right) \odot H^l + Z^l \odot X^{l-1}, \quad & l=1,\dots,L, \\
\text{output}\quad
    f(x;\theta) & = \left(x - K \ue^{-rt} \right)^+ + \sigma_2 \left( W X^L + b \right). 
\end{align*}
Here, $L$ is the number of hidden layers, $\sigma_i$ is the activation function for $i=1,2$, and $\odot$ denotes the element-wise multiplication. 
Moreover, let us point here that $\textbf{x}$ denotes the vector of inputs, which are the asset price and the variances (in the stochastic volatility models), and then $\mathbf{x}_1=x$.
In the numerical experiments, we have used 3 layers and 50 neurons per layer. 
The activation functions we have selected are the hyperbolic tangent function, $\sigma_1(x) = \tanh(x)$, and the softplus function, $\sigma_2(x) = \log \left( \ue^x +1 \right)$, which guarantees that the option price remains above the no-arbitrage bound. 

We consider a maturity of $T=1.0$ year, and use the parameters sets for the Heston and the lifted Heston model from \citet[Table 1]{heston1993closed} and \citet[page 321]{abi2019multifactor}.
We set the number of time steps equal to $N = 100$, use $2000$ sampling stages in each time step, while in each sampling stage we take 600 samples per dimension, \textit{i.e.} $600(n+1)$ samples, following the recommendations in \citet{georgoulis2023discrete}.
Here, $n$ denotes the number of variance processes, in the case of stochastic volatility models.
As for the sampling domain, we consider the moneyness $x \in [0.01,3.0]$, the Heston volatility $V \in [0.001, 0.1]$ and the lifted Heston volatilities $V^{n,i} \in \left[ -V_i^{\text{high}}, V_i^{\text{high}} \right]$, with
\[
    V_i^{\text{high}} = 3 \sqrt{\frac{\eta^2 V_0}{2 \gamma_i^n} \left( 1 - \ue^{-2 \gamma_i^n T} \right)}.
\]
After sampling $V^{n,i}$ uniformly, we calculate $V^n$ and discard all combinations that result in a negative $V^n$. 
Moreover, we set $x_p = 2.0$ 
In the optimization stage, we use the Adam algorithm \cite{kingma2014adam} with a learning rate $\alpha = 3 \times 10^{-4}$, $(\beta_1,\beta_2) = (0.9,0.999)$ and zero weight decay. 
Finally, the training was performed on the DelftBlue supercomputer \cite{DHPC2022}, using a single NVidia Tesla V100S GPU.


\section{Numerical results}
\label{sec:results}

Let us now present and discuss the outcome of certain numerical experiments involving all models outlined in Section \ref{sec:examples}. 
We are interested in both the accuracy, in Subsection \ref{sec:accuracy}, and the speed, in Subsection \ref{sec:running}, of the numerical method developed in this article. 

We will compare the TDGF method with another popular deep learning method for solving PDEs, the DGM of \citet{sirignano2018dgm}. 
In the DGM approach, the solution of the PDE is translated into a quadratic minimization problem, which is minimized using stochastic gradient descent. 
In our setting, this minimization takes the form:
\[
\begin{aligned}
\big \Vert u_t  - \nabla \cdot \left( A \nabla u \right) + \mathbf{b} \cdot \nabla u + ru \big \Vert_{L^2([0,T] \times \Omega)}^2 
+ \big \Vert u(0,\mathbf{x}) - \Psi(x) \big \Vert_{L^2(\Omega)}^2 \to \min. 
\end{aligned}
\]
Let us recall that, for large levels of asset price / moneyness, the option price grows linearly with the asset price / moneyness. 
Therefore, we add a loss term punishing the derivative of the option price with respect to the stock deviating from 1 at $x=3.0$. 

In order to ensure a fair comparison between the two methods, we use 200,000 sampling stages for the DGM. 
In each stage, we use $600(n+2)$ samples for the PDE in the interior domain and $600(n+1)$ samples for the initial and boundary conditions. 
Moreover, we use the same network architecture as for the TDGF with the same number of layers and neurons per layer, and also use the Adam optimizer.

In order to compare the accuracy of the two methods, we need a reference value. 
In the Black--Scholes model, in the case of a European call option, the option pricing PDE has an exact solution:
\[
u(t, x) = x \Phi(d_1) - K \ue^{-r \tau} \Phi(d_2),
\]
with $\Phi$ the cumulative standard normal distribution function, $\tau = T-t$ the time to maturity,
\[
d_1 = \frac{\log \left( \frac{x}{K} \right) + \left( r + \frac{\sigma^2}{2} \right) \tau}{\sigma \sqrt{\tau}} \quad \textrm{and} \quad d_2 = d_1 - \sigma \sqrt{\tau}.
\]

In the Heston model, the option pricing PDE does not admit an analytical solution. 
However, the model is affine and the characteristic function does have an analytical expression; see \citet{heston1993closed}. 
Using the characteristic function, we compute the reference price with the COS method of \citet{fang2009novel}, with $N=512$ terms and truncation domain $\left[-10 \tau^{\frac{1}{4}}, 10 \tau^{\frac{1}{4}} \right]$.

Finally, in the lifted Heston model, the characteristic function does not have an analytical expression but, as this model is again affine, the characteristic function is known up to the solution of a system of Riccati equations; see \citet{abi2019lifting}. 
We solve these equations using an implicit-explicit discretization scheme with 500 time steps, after which we apply the COS method again to compute the reference price.


\subsection{Accuracy}
\label{sec:accuracy}

In the Black--Scholes model, we have only one input variable, the asset price $S$. 
Figure \ref{fig:BS_call_error} presents the relative $L^2$- and the absolute maximum error for the TDGF and the DGM methods in the Black--Scholes model. 
These errors are calculated over 47 evenly spaced grid points in the interval $[0.01,3]$. 
Both methods have a comparable (small) error.

\begin{figure}
    \centering
    \begin{subfigure}{0.49\textwidth}
    \centering
    \includegraphics[width=\linewidth]{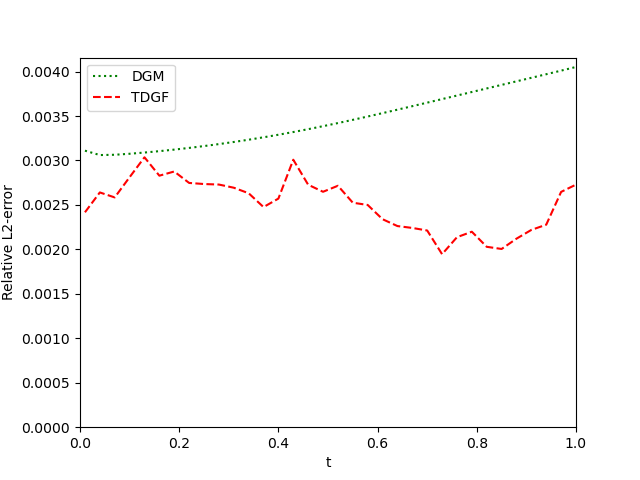}
    \caption{Relative $L^2$-error.}
    \end{subfigure}
    \begin{subfigure}{0.49\textwidth}
    \centering
    \includegraphics[width=\linewidth]{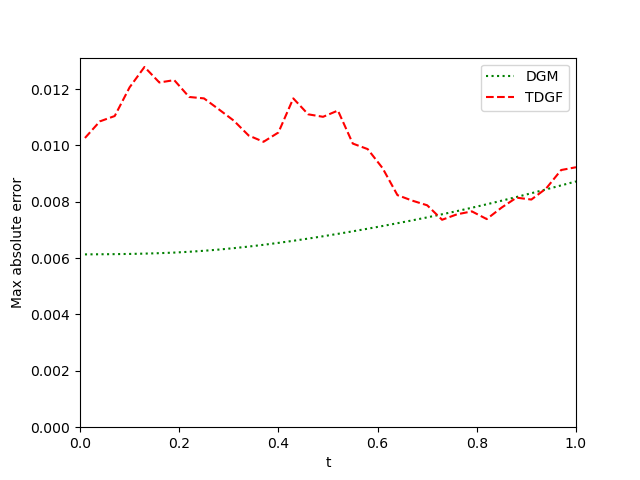}
    \caption{Maximal absolute error.}
    \end{subfigure}
    \caption{Errors of the two methods in the Black--Scholes model against time, with interest rate $r=0.05$ and volatility $\sigma = 0.25$.}
    \label{fig:BS_call_error}
\end{figure}

In the Heston model, we have two input variables, the asset price $S$ and the variance $V$. 
Figure \ref{fig:Heston_call_error} presents the relative $L^2$- and the absolute maximum error for the TDGF and the DGM methods in the Heston model. 
These figures are presented for fixed $V_0 = 0.03$, however other values of $V_0$ between 0.01 and 0.09 provide similar results. 
The DGM is more accurate than the TDGF method, although the TDGF is just as accurate as in the Black--Scholes model, and the error is any case small, in the order of $10^{-3}$.


\begin{figure}
    \centering
    \begin{subfigure}{0.49\textwidth}
    \centering
    \includegraphics[width=\linewidth]{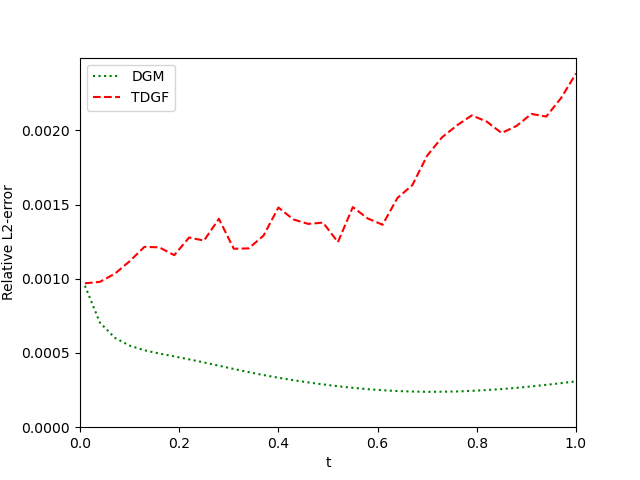}
    \caption{Relative $L^2$-error}
    \end{subfigure}
    \begin{subfigure}{0.49\textwidth}
    \centering
    \includegraphics[width=\linewidth]{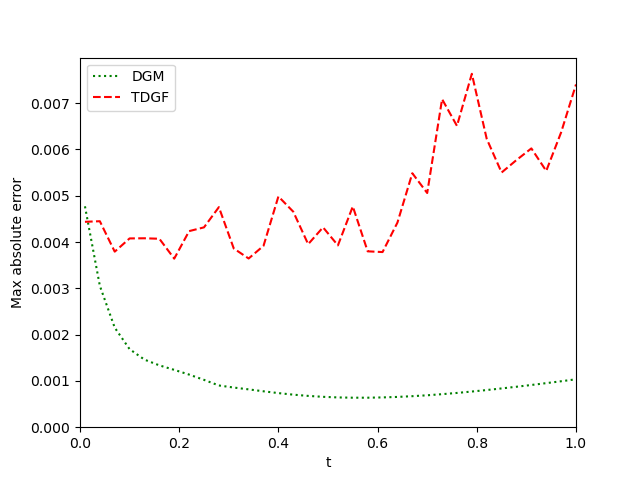}
    \caption{Maximal absolute error}
    \end{subfigure}
    \caption{Error of the two methods in the Heston model against time, with $r=0.0$ and $\eta = 0.1$, $\rho=0.0$, $\kappa=0.01$, $V_0 = 0.03$ and $\lambda =2.0$.}
    \label{fig:Heston_call_error}
\end{figure}

In the lifted Heston model, we have $n+1$ input variables, the asset price $S$ and the variances $V^{n,i}$, with $i=1,\dots,n$. 
Figure \ref{fig:Heston/rH_call_n=1_error} presents the relative $L^2$- and the maximum absolute error for the TDGF and the DGM methods in this model, in the case $n=1$, using the same parameters as in the Heston model before. 
The numerical results are comparable for the two methods, and both methods even seem to perform slightly better than in the previous example.

\begin{figure}
    \centering
    \begin{subfigure}{0.49\textwidth}
    \centering
    \includegraphics[width=\linewidth]{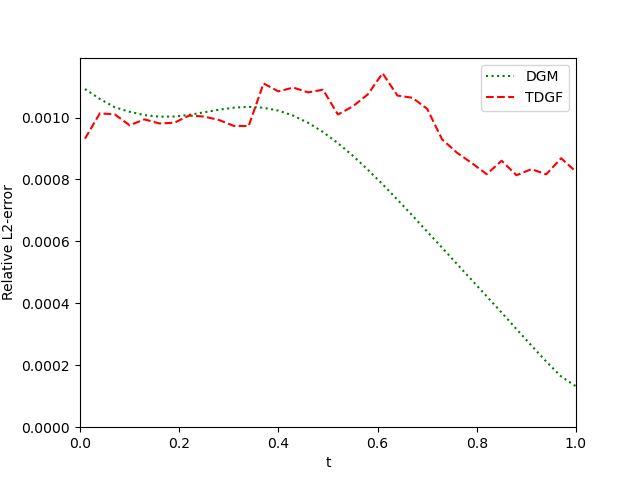}
    \caption{Relative $L^2$-error}
    \end{subfigure}
    \begin{subfigure}{0.49\textwidth}
    \centering
    \includegraphics[width=\linewidth]{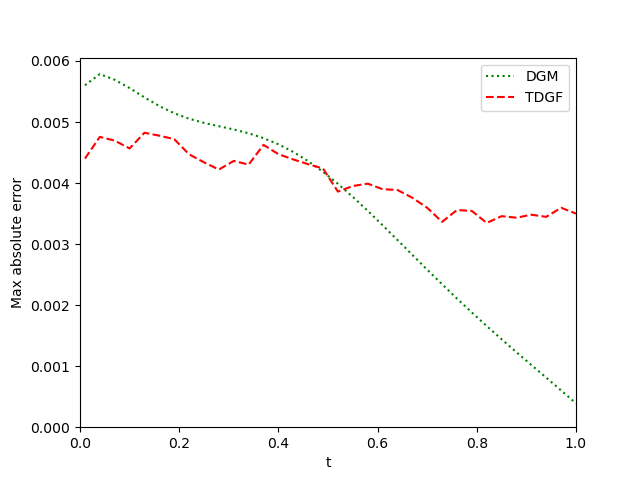}
    \caption{Maximal absolute error}
    \end{subfigure}
    \caption{Errors of the two methods in the lifted Heston model with $n=1$ variance process against time, with $r=0.0$, $\eta=0.1$, $\rho=0.0$, $\kappa=0.01$, $V_0=0.01$ and $\lambda=2.0$.}
    \label{fig:Heston/rH_call_n=1_error}
\end{figure}

Figure \ref{fig:Heston/rH_call_n=20_error} presents the relative $L^2$- and the maximum absolute error for the TDGF and the DGM methods again in the lifted Heston model, for the same parameter set as previously, but now with $n=20$ variance factors. 
The errors of both methods are comparable and, although they have increased compared to the $n=1$ case, they are still low, in the order of $10^{-3}$.


\begin{figure}
    \centering
    \begin{subfigure}{0.49\textwidth}
    \centering
    \includegraphics[width=\linewidth]{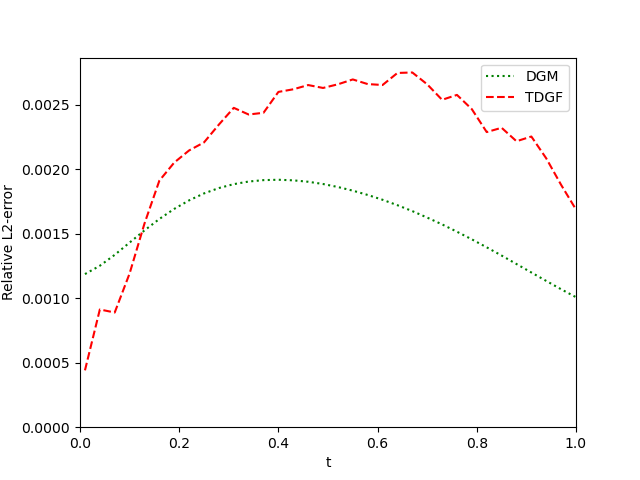}
    \caption{Relative $L^2$-error}
    \end{subfigure}
    \begin{subfigure}{0.49\textwidth}
    \centering
    \includegraphics[width=\linewidth]{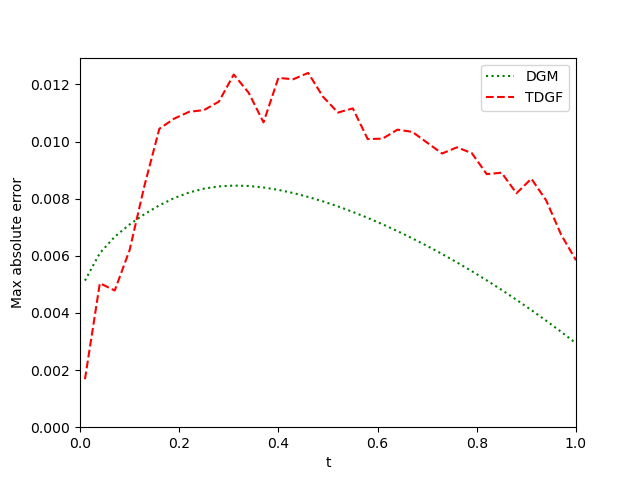}
    \caption{Maximal absolute error}
    \end{subfigure}
    \caption{Errors of the two methods in the lifted Heston model with $n=20$ variance process against time, with $r=0.0$, $\eta=0.1$, $\rho=0.0$, $\kappa=0.01$, $V_0=0.01$ and $\lambda=2.0$.}
    \label{fig:Heston/rH_call_n=20_error}
\end{figure}

Figure \ref{fig:AJ/rH_call_n=1_error} presents the relative $L^2$- and the absolute maximum error for the two methods in the lifted Heston model, now for the parameter set suggested by \citet{abi2019lifting} with $n=1$. 
The main difference here is that the correlation between the asset price and the variance processes is no longer zero.
The error in both methods has slightly increased with respect to the previous example, with the DGM performing slightly better than the TDGF. 
However, the overall error is still small, in the order of $10^{-3}$.


\begin{figure}
    \centering
    \begin{subfigure}{0.45\textwidth}
    \centering
    \includegraphics[width=\linewidth]{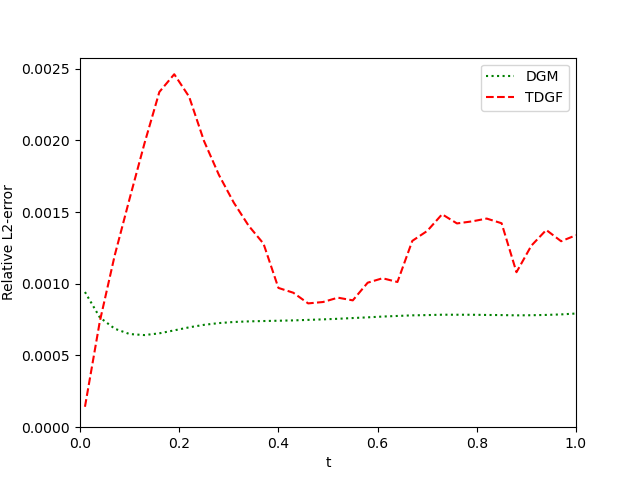}
    \caption{Relative $L^2$-error}
    \end{subfigure}
    \begin{subfigure}{0.45\textwidth}
    \centering
    \includegraphics[width=\linewidth]{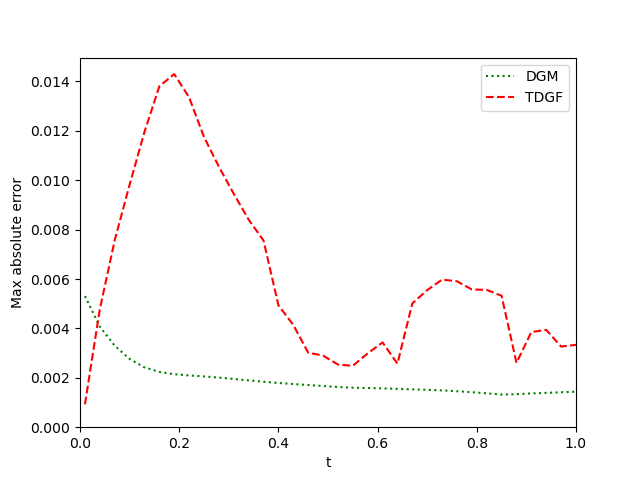}
    \caption{Maximal absolute error}
    \end{subfigure}
    \caption{Errors of the two methods in the lifted Heston model with $n=1$ variance process against time, with $r=0.0$, $\eta=0.3$, $\rho=-0.7$, $\kappa=0.02$, $V_0=0.02$ and $\lambda=0.3$.}
    \label{fig:AJ/rH_call_n=1_error}
\end{figure}

Finally, 
Figure \ref{fig:AJ/rH_call_n=5_error} presents the relative $L^2$- and the absolute maximum error for the two methods in the lifted Heston model, for the same parameter set as in the previous two figures, but now with $n=5$ and $\rho=-0.2$. 
The combination of increased dimension (\textit{i.e.} 5 variance processes compared to one before) and decreased correlation results in the errors of both methods being similar to the previous case, and still small, in the order of $10^{-3}$.
However, as the number of variance processes increase while the correlation is non-zero, then the errors in both methods are increasing rapidly, and we can no longer accommodate \textit{e.g.} 20 variance factors as we did in the uncorrelated case.


\begin{figure}
   \begin{subfigure}{0.45\textwidth}
    \centering
    \includegraphics[width=\linewidth]{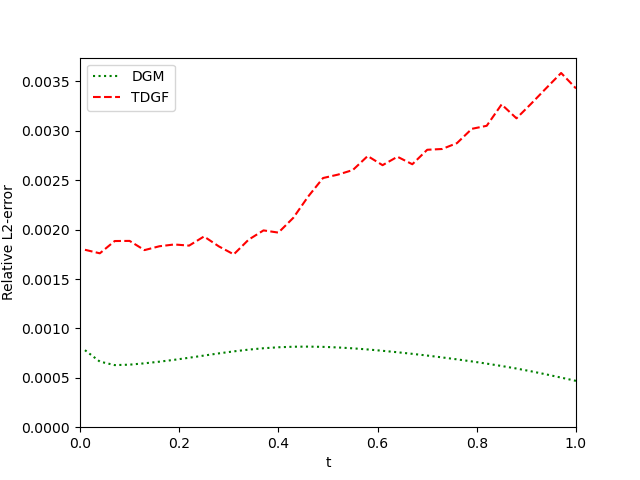}
    \caption{Relative $L^2$-error}
    \end{subfigure}
    \begin{subfigure}{0.45\textwidth}
    \centering
    \includegraphics[width=\linewidth]{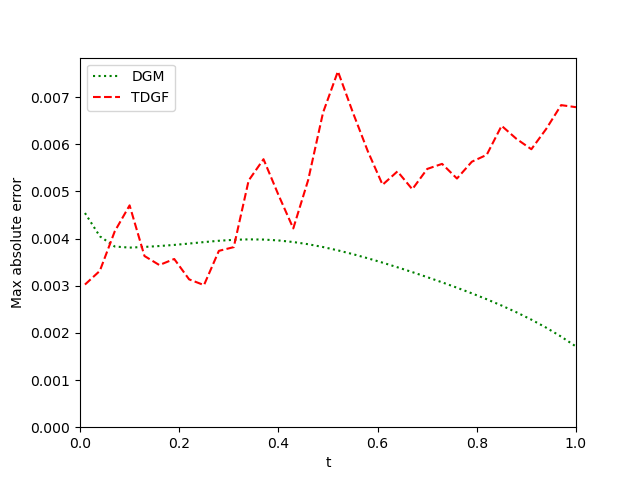}
    \caption{Maximal absolute error}
    \end{subfigure}
    \caption{Errors of the two methods in the lifted Heston model with $n=5$ variance processes against time, with $r=0.0$, $\eta=0.3$, $\rho=-0.2$, $\kappa=0.02$, $V_0=0.02$ and $\lambda=0.3$.}
    \label{fig:AJ/rH_call_n=5_error}
\end{figure}


\subsection{Training and computational times}
\label{sec:running}

Table \ref{tab:training_time} summarizes the training times for the TDGF and the DGM methods in the different models. 
As we have expected, due to the time stepping and the absence of a second derivative in the cost function, the training of the TDGF method is faster than for the DGM method.
Moreover, the training time for the TDGF method does not increase as the dimension of the problem (\textit{i.e.} the number of variance processes) increases.
On the contrary, the training time for the DGM method increases significantly with the number of dimensions, and becomes more than 4 times slower in the lifted Heston model when the dimension of the variance processes increases from 1 to 20.

The computational times for the TDGF and the DGM methods in all models are presented in Table \ref{tab:computing_time}. 
The computational times are taken as the average over 34 time points of computing the option prices on a grid of 47 moneyness points. 
Even though we take an average, there is still a high variance in the computational times. 
Both methods are significantly faster than the COS method in the lifted Heston model, in which we do not know the characteristic function explicitly and a system of Riccati equations needs to be solved for every evaluation of the characteristic function. 

The possibility of splitting the computation of an option price in a slow, offline phase (training) and a very fast, online one (computation), is one of the advantages of deep learning methods such as the TDGF proposed here and the DGM. 
However, we should point out that when changing the model parameters, then the TDGF and the DGM need to be retrained, thus the overall time (training and computation) is slower than the evaluation of the COS method.
The development of ``deep parametric PDE" methods as in \citet{Glau_Wunderlich_2022} would be an interesting avenue to explore in that respect.
Finally, let us also mention that the TDGF requires more memory than the DGM, around 10-100MB, as it has to store a separate neural network for each time step, whereas the DGM needs only one neural network for all times; the increase in speed though justifies this cost. 

\begin{table}
    \centering
    \begin{tabu}{l|c|c|c|c|c} \hline
        Model & Black--Scholes & Heston & $\ell$-Heston, $n=1$ & $\ell$-Heston, $n=5$ & $\ell$-Heston, $n=20$ \\ \hline
        DGM & {$7.0 \times 10^3$} & {$12.2 \times 10^3$} & {$12.5 \times 10^3$} & {$21.0 \times10^3$} & {$54.3 \times10^3$} \\ \hline
        TDGF & {$3.8 \times 10^3$} & {$\phantom{0}5.8 \times 10^3$} & {$\phantom{0}5.9 \times 10^3$} & {$\phantom{0}6.3 \times 10^3$} & {$\phantom{0}7.4 \times 10^3$} \\ \hline
    \end{tabu}
    \caption{Training time in seconds of the different methods for a European call option in the different models.}
    \label{tab:training_time}
\end{table}

\begin{table}
    \centering
    \begin{tabu}{l|l|l|l|l|l} \hline
        Model & Black--Scholes & Heston & $\ell$-Heston, $n=1$ & $\ell$-Heston, $n=5$ & $\ell$-Heston, $n=20$ \\ \hline
        Exact/COS & {0.0015} & {0.013} & {9.1} & {9.5} & {10.0} \\ \hline
        DGM & {0.0086} & {0.0015} & {0.0043} & {0.0052} & {\phantom{1}0.0015} \\ \hline
        TDGF & {0.0018} & {0.0018} & {0.0019} & {0.0019} & {\phantom{1}0.0018} \\ \hline
    \end{tabu}
    \caption{Computational time in seconds of the different methods for a European call option in the different models.}
    \label{tab:computing_time}
\end{table}


\section{Conclusion}
\label{sec:conclusion}

We have developed in this article a novel deep learning method for option pricing in diffusion models.
Starting from the option pricing PDE, using a backward differentiation time-stepping scheme and results from the calculus of variations, we can rewrite the PDE as an energy minimization problem for a suitable energy functional.
The time-stepping provides a good initial guess for the next step in the optimization procedure, while the energy formulation yields an equation that only has first derivatives, instead of second ones.
The energy minimization problem offers a suitable candidate for a cost function, and the PDE is solved by training a neural network using stochastic gradient descent-type optimizers (\textit{e.g.} Adam).
We have considered as examples for our method the Black--Scholes model, the Heston model and the lifted Heston model, and derived the appropriate representations for the PDE in all these cases.

Overall, the TDGF method developed here performs well in the numerical experiments, with an error in the order of $10^{-3}$.
Moreover, the training and computation (evaluation) times for the TDGF method are faster and more consistent compared to the popular DGM method.
However, when the dimension of the variance processes in the lifted Heston model increases and the correlation increases simultaneously, then the error also increases.
This increase in the error is due to the diffusion term becoming smaller, which means that the explicit term in the decomposition \eqref{eq:discretization} becomes the dominant term.
The development of suitable numerical methods for this regime will be the subject of future research.




\bibliographystyle{abbrvnat}
\bibliography{bibliography}


\end{document}